\documentclass[draftclsnofoot, onecolumn, 12pt]{IEEEtran}

\usepackage{amsthm,amsmath,amssymb,graphicx,paralist,subfigure,pdfpages,cite}
\newtheorem{lem}{Lemma} 
\newtheorem{thm}{Theorem}
\newtheorem{definition}{Definition}

\def\ln{{\rm ln}}

\def\mc{\mathcal}
\def\mb{\mathbf}
\def\mbb{\mathbb}
\def\ra{\rightarrow}

\IEEEoverridecommandlockouts

\begin{document}
\title{Asymptotic stability of stochastic LTV systems with applications to distributed dynamic fusion}
\author{Sam Safavi$^\dagger$,~\emph{Student Member,~IEEE}, and Usman A. Khan$^\dagger$,~\emph{Senior Member,~IEEE}\thanks{$^\dagger$The authors are with the Department of Electrical and Computer Engineering, Tufts University, 161 College Ave, Medford, MA 02155, {\texttt{\{sam248,khan\}@ece.tufts.edu}}. This work is partially supported by an NSF Career award: CCF \# 1350264.}}
\maketitle
\thispagestyle{empty}

\begin{abstract}
In this paper, we investigate asymptotic stability of linear time-varying  systems with (sub-) stochastic system matrices. Motivated by distributed dynamic fusion over networks of mobile agents, we impose some mild regularity conditions on the elements of time-varying system matrices. We provide sufficient conditions under which the asymptotic stability of the LTV system can be guaranteed. By introducing the notion of slices, as non-overlapping partitions of the sequence of systems matrices, we obtain stability conditions in terms of the slice lengths and some network parameters. In addition, we apply the LTV stability results to the distributed leader-follower algorithm, and show the corresponding convergence and steady-state. An illustrative example is also included to validate the effectiveness of our approach.
\end{abstract}

\section{introduction}\label{intro}
Stability of linear time-varying (LTV) systems has been a topic of significant interest in a wide range of disciplines including but not restricting to mathematical modeling and control of dynamical systems,~\cite{rosenbrook1963stability,1100529,1084637,Ilchmann1987157,tsakalis1993linear,DaCunha2005381,Phat2006343}. Discrete-time, LTV dynamics can be represented by the following model:
\begin{equation}\label{1}
{\bf{x}}(k+1)={{P}}_k{\bf{x}}(k)+B_k\mb{u}_k,\qquad k \geq 0,
\end{equation}
where~$\mb{x}_k\in\mbb{R}^n$ is the state vector,~$P_k$'s are the system matrices,~$B_k$'s are the input matrices, and~$\mb{u}_k\in\mbb{R}^s$ is the input vector. This model is particularly relevant to design and analysis of distributed fusion algorithms when the system matrices,~$P_k$'s, are (sub-) stochastic, i.e. they are non-negative and each row sums to at most~$1$. Examples include leader-follower algorithms,~\cite{tanner02,4200874}, consensus-based control algorithms,~\cite{5509836,5456181,4456762}, and sensor localization,~\cite{khan2009distributed,khan2010diland}. 

In contrast to the case when the system matrices,~$P_k$'s, are time-invariant, i.e.~$P_k=P,\forall k$, as in many studies related to the above examples, we are motivated by the scenarios when these system matrices are time-varying. The dynamic system matrices do not only model time-varying neighboring interactions, but, in addition, capture agent mobility in multi-agent networks. Consider, for example, the leader-follower algorithm,~\cite{tanner02,4200874}, where~$n$ \emph{sensors} update their states,~$\mb{x}_k$'s in Eq.~\eqref{1}, as a linear-convex combination of the neighboring states, and~$s=1$ \emph{anchor} keeps its (scalar) state,~$u_k$, fixed at all times. It is well-known that under mild conditions on network connectivity the sensor states converge to the anchor state. However, the neighboring interactions change over time if the sensors are  mobile. In the case of possibly random motion over the sensors, at each time~$k$, it is not guaranteed that a sensor can find any neighbor at all. If a sensor finds a set of neighbors to exchange information, none of these neighbors may be an anchor. We refer to the general class of such time-varying fusion algorithms over mobile agents as \emph{Distributed Dynamic Fusion~(DDF)}. In this context, we study the conditions required on the DDF system matrices such that the \emph{dynamic} fusion converges to (a linear combination of) the anchor state(s).

For linear time-invariant (LTI) systems, a necessary and sufficient condition for stability is that the \textit{spectral radius}, i.e. the absolute value of the largest eigenvalue, of the system matrix is subunit. A well-known result from the matrix theory is that if the (time-invariant) system matrix,~$P$, is irreducible \textit{and} sub-stochastic, sometimes referred to as \emph{uniformly sub-stochastic},~\cite{kolpakov,kolpakov_rus:83}, the spectral radius of~$P$ is strictly less than one and~$\lim_{k \rightarrow \infty}\mb{x}_k$ converges to zero. In contrast, the DDF algorithms over mobile agents result into a time-varying system, Eq.~\eqref{1}, where a system matrix,~$P_k$, at any time~$k$ is non-negative, and can be: (i) identity if no sensor is able to update its state; (ii) stochastic if the updating sensor divides the total weight of~$1$ among the sensors in its neighborhood; or, (iii) sub-stochastic if the total weight of~$1$ is divided among both sensors and anchors. In addition, it can be verified that in DDF algorithms, the resulting LTV system may be such that the spectral radius,~$\rho(P_k)$, of the system matrices follow~$\rho(P_k)=1,\forall k$. This is, for example, when only a few sensors update and the remaining stick to their past states. 
 
Asymptotic stability for LTV systems may be characterized by the \textit{joint spectral radius} of the associated family of system matrices. Given a finite set of matrices,~$\mathcal{M}=\{{A}_{1},\ldots,{A}_{m}\}$, the joint spectral radius of the set~$\mathcal{M}$, was introduced by Rota and Strang,~\cite{rota1960note}, as a generalization of the classical notion of spectral radius, with the following definition:
\begin{equation*}\label{2}
\rho(\mathcal{M}):=\lim_{k \rightarrow \infty} \max\limits_{A \in {\mathcal{M}}_k } {\Vert {A}\Vert}^{\frac{1}{k}},
\end{equation*}
in which~${\mathcal{M}}_k$ is the set of all possible products of the length~$k \geq 1$, i.e.
\begin{equation*}
{\mathcal{M}}_k =\{A_{i_1} A_{i_2} \ldots A_{i_k} \}: 1 \leq i_j \leq m,~j=1,\ldots,k.
\end{equation*}
Joint spectral radius (JSR) is independent of the choice of norm, and represents the maximum growth rate that can be achieved by forming arbitrary long products of the matrices taken from the set~${\mathcal{M}}$. It turns out that the asymptotic stability of the LTV systems, with system matrices taken from the set~$\mathcal{M}$, is guaranteed,~\cite{parrilo2008approximation}, if and only if
\begin{equation*}
\rho(\mathcal{M}) < 1.
\end{equation*}
  
Although the JSR characterizes the stability of LTV systems, its computation is NP-hard,~\cite{tsitsiklis1997lyapunov}, and the determination of a strict bound is undecidable,~\cite{blondel2000boundedness}. Naturally, much of the existing literature has focused on JSR approximations,~\cite{parrilo2008approximation,gripenberg1996computing,blondel2000boundedness,jungers2009joint,blondel2005computationally,tsitsiklis1997lyapunov,qu2005products,touri2011product}. For example, Ref.~\cite{blondel2005computationally} studies lifting techniques  to approximate the JSR of a set of matrices. The main idea is to build a lifted set with a larger number of matrices, or a set of matrices with higher dimensions, such that the relation between the JSR of the new set and the original set is known. Lifting techniques provide better bounds at the price of a higher computational cost. In~\cite{parrilo2008approximation}, a sum of squares programming technique is used  to approximate the JSR of a set of matrices; a bound on the quality of the approximation is also provided, which is independent of the number of matrices. Stability of LTV systems is also closely related to the convergence of infinite products of matrices. Of particular interest is the special case of the (infinite) product of non-negative and/or (sub-) stochastic matrices, see~\cite{guu2003convergence,daubechies1992sets,bru1994convergence,beyn1997infinite,hartfiel1974infinite,pullman1966infinite,kochkarev1995continuous,Elsner1997133}. In addition to non-negativity and sub-stochasticity, the majority of these works set other restrictions, such as irreducibility or bounds on the row sum on each matrix in the set. 

The main contributions of this paper are as follows. \emph{Design}: we provide a set of conditions on the elements of the system matrices under which the asymptotic stability of the corresponding LTV system can be guaranteed. \emph{Analysis}: we propose a general framework to determine the stability of an (infinite) product of (sub-) stochastic matrices. Our approach does not require either the computation or an approximation of the JSR. Instead, we partition the infinite set of system matrices (stochastic, sub-stochastic, or identity) into non-overlapping slices--a slice is defined as the smallest product of (consecutive) system matrices such that: (i) every row in a slice is strictly less than one; and, (ii) the slices cover the entire sequence of system matrices. Under the conditions established in the design, we subsequently show that the infinity norm of each slice is subunit (recall that in the DDF setup, infinity norm of each system matrix is one). Finally, in order to establish the relevance to the fusion \textit{applications} of interest, we use the theoretical results to derive the convergence and steady-state of a dynamic leader-follower algorithm. 

An important aspect of our analysis lies in the study of slice lengths. First, we show that longer slices may have an infinity norm that is closer to one as compared to shorter slices. Clearly, if one can show that each slice norm is subunit (with a uniform upper bound of~$<1$) then one further has to guarantee an infinite number of such slices to ensure stability. The aforementioned argument naturally requires slices of finite length, as finite slices covering infinite (system) matrices lead to an infinite number of slices. An avid reader may note that guaranteeing a sharp upper bound on the length of every slice may not be possible for certain network configurations. To address such configurations, we characterize the rate at which the slices (not necessarily in an order) grow large such that the LTV stability is not disturbed. In other words, a longer slice may capture a slow information propagation in the network; characterizing the aforementioned growth is equivalent to deriving the rate at which the information propagation may deteriorate in a network such that the fusion is still achievable. 

The rest of this paper is organized as follows. We formulate the problem in Section~\ref{PF}, while Section~\ref{IP} studies the convergence of an infinite product of (sub-) stochastic matrices. Stability of discrete-time LTV systems with (sub-) stochastic system matrices is studied in Section~\ref{stability}. We provide applications to distributed dynamic fusion in Section~\ref{app} and illustrations of the results in Section~\ref{example}. Finally, Section~\ref{conc} concludes the paper.

\section{Problem formulation}\label{PF}
In this paper, we study the \textit{asymptotic stability} of the following Linear Time-Varying (LTV) dynamics:
\begin{equation}\label{eq2}
{\bf{x}}(k+1)={{P}}_k{\bf{x}}(k)+{{B}}_k{\bf{u}}(k),\qquad k \geq 0,
\end{equation}
where~${\mb{x}}(k)\in\mbb{R}^n$ is the state vector,~${{P}}_k\in\mbb{R}^{n \times n}$ is the time-varying system matrix,~${B}_k\in\mbb{R}^{n \times s}$ is the time-varying input matrix,~${\bf{u}}(k)\in\mbb{R}^s$ is the input vector, and~$k$ is the discrete-time index. We consider the system matrix,~${P}_k$ at each~$k$, to be non-negative and either \emph{sub-stochastic}, \emph{stochastic}, or \emph{identity}, along with some  conditions on its elements. The input matrix,~${B}_k$ at each~$k$, may be arbitrary as long as some regularity conditions are satisfied. These regularity conditions on the system matrices,~~${P}_k$'s and~~${B}_k$'s, are collected in the Assumptions {\bf A0--A2} in the following. 

In this paper, we are interested in deriving the conditions on the corresponding system matrices under which the LTV dynamics in Eq.~\eqref{eq2} forget the initial condition,~$\mb{x}(0)$, and converge to some function of the input vector,~$\mb{u}_k$. The motivation behind this investigation can be cast in the context of distributed fusion over dynamic graphs that we introduce in the following. 

\subsection{Distributed Dynamic Fusion} 
Consider a network of~$n+s$ mobile nodes moving arbitrarily in a (finite) region of interest, where~$n$ mobile sensors implement a distributed algorithm to obtain some relevant function of~$s$ (mobile) anchors; examples include the leader-follower setup,~\cite{tanner02,4200874}, and sensor localization,~\cite{khan2009distributed,khan2010diland}. The sensors may be thought of as mobile agents that collect information from the anchors and disseminate within the sensor network. Each node may have restricted mobility in its respective region and thus many sensors may not be able to directly connect to the anchors. Since the motion of each node is arbitrary, the network configuration at any time~$k$ is completely unpredictable. It is further likely that at many time instants, no node has any neighbor in its communication radius. 

Formally, sensors, in the set~$\Omega$, are the nodes in the graph that update their states,~$x_i(k)\in\mbb{R}, i=1,\ldots,n$, as a linear-convex function of the neighboring nodes; while anchors, in the set~$\kappa$, are the nodes that inject information,~$u_j(k)\in\mbb{R},j=1,\ldots,s$, in the network. Let~$\mc{N}_i(k)$ denote the set of neighbors (not including sensor~$i$) of sensor~$i$ according to the underlying graph at time~$k$, with~$\mc{D}_i(k)\triangleq\{i\}\cup\mc{N}_i(k)$. We assume that at each time~$k$, only one sensor, say~$i$, updates its state\footnote{Although multiple sensors may update their states at each iteration, without loss of generality, we assume that at most one sensor may update.},~$x_i(k)$. Since the underlying graph is dynamic, the updating sensor~$i$ implements one of the following updates:
\begin{enumerate}[(i)]
\item No neighbors:
\begin{eqnarray}\label{U1}
x_i(k+1) = x_i(k),\qquad\mc{N}_i(k)=\emptyset.
\end{eqnarray}
\item No neighboring anchor,~$\mc{N}_i(k)\cap\kappa=\emptyset$:
\begin{eqnarray}\label{U2}
x_i(k+1) = \sum\limits_{l\in{\mc{D}_{i}(k)}}({P}_{k})_{i,l}{x}_{l}(k). 
\end{eqnarray}
\item At least one anchor as a neighbor:
\begin{eqnarray}\nonumber
x_i(k+1) &=& 
\sum\limits_{l\in{\mc{D}_{i}(k)\cap\Omega}}({P}_{k})_{i,l}{x}_{l}(k)\\ \label{U3}
&+& \sum\limits_{j\in{\mc{D}_i(k)}\cap\kappa}({B}_{k})_{i,j}{u}_{j}(k),
\end{eqnarray}
with~$\mc{N}_i(k)\cap\kappa\neq\emptyset$.
\end{enumerate}
At every other (non-updating) sensor,~$l\neq i$, we have
\begin{eqnarray}\label{U4}
x_l(k+1) = x_l(k).
\end{eqnarray}

\subsection{Assumptions} 
Let~${P}_k=({P}_{k})_{i,l}$, and~${B}_k=({B}_{k})_{i,j}$, we now enlist the assumptions:

\noindent {\bf A0}: When the updating sensor,~$i$, has no anchor as a neighbor, the update in Eq.~\eqref{U2} is linear-convex, i.e.
\begin{eqnarray}
\sum\limits_{l\in{\mc{D}_{i}(k)}\cap\Omega}({P}_{k})_{i,l} = 1,
\end{eqnarray}
resulting in a (row) stochastic system matrix,~${P}_k$.

\noindent {\bf A1}: When the updating sensor,~$i$, has no anchor but at least one sensor as a neighbor, the weight it assigns to each neighbor (including the self-weight) is such that
\begin{eqnarray}\label{bnd1}
0 < \beta_1 \leq ({P}_{k})_{i,l} < 1,\qquad \forall l \in {\mathcal{D}}_i(k),\beta_1\in\mbb{R},
\end{eqnarray}

\noindent {\bf A2}: When the updating sensor updates with an anchor, the update, Eq.~\eqref{U3}, over the sensors,~$\mc{D}_i(k)\cap\Omega$, satisfies
\begin{eqnarray}\label{bnd2}
\sum\limits_{l\in{\mc{D}_{i}(k)}\cap\Omega}({P}_{k})_{i,l} \leq\beta_2 < 1,
\end{eqnarray}
resulting in a sub-stochastic system matrix,~${P}_k$. Also note that the update over the anchors,~$\mc{N}_i(k)\cap\kappa$, in Eq.~\eqref{U3}, follows
\begin{eqnarray}\label{bnd3}
({B}_{k})_{i,j}\geq\alpha>0,\qquad\forall j\in\mc{N}_i(k)\cap\kappa.
\end{eqnarray}
If, in addition, we enforce~$\sum_l ({P}_{k})_{i,l}+\sum_j ({B}_{k})_{i,j}=1$, as it is assumed in leader-follower,~\cite{tanner02,4200874}, or sensor localization,~\cite{khan2009distributed,khan2010diland}, Eq.~\eqref{bnd3} naturally leads to the bound in Eq.~\eqref{bnd2}.

Clearly, which of the four updates in Eqs.~\eqref{U1}--\eqref{U4} is applied by the updating sensor,~$i$, depends on being able to satisfy the corresponding assumptions ({\bf A0--A2}), \emph{in addition} to the neighborhood configuration. Indeed, letting 
\begin{eqnarray*}
\mb{x}(k) &=& \left[x_1(k),\ldots,x_n(k)\right]^\top,\\
\mb{u}(k) &=& \left[u_1(k),\ldots,u_m(k)\right]^\top,
\end{eqnarray*} 
result into the LTV system in Eq.~\eqref{eq2}. Clearly, the time-varying system matrices,~${P}_k$, are either sub-stochastic, stochastic, or identity, depending on the nature of the update. 

{\bf Remarks:} It is meaningful to comment on the assumptions made above. Non-negativity and stochasticity are standard in the literature concerning relevant iterative algorithms and multi-agent fusion, see e.g.~\cite{5509836,5456181,4456762,khan2009distributed}. When there is a neighboring anchor, Eq.~\eqref{bnd2} provides an \emph{upper bound on unreliability} thus restricting the amount of unreliable information added in the network by a sensor. Eq.~\eqref{bnd3}, on the other hand, can be viewed as a \emph{lower bound on reliability}; it ensures that whenever an anchor is included in the update, a certain amount of information is always contributed by the anchor. An avid reader may note that Eq.~\eqref{bnd3} guarantees that the following does not occur:~$({B}_{k_1})_{i,j}\ra0$, where~$k_1\geq0$ is a subsequence within~$k$ denoting the instants when Eq.~\eqref{U3} is implemented. Similarly, Eqs.~\eqref{bnd1} and~\eqref{bnd2} ensure that no sensor is assigned a weight arbitrarily close to~$1$ and thus no sensor may be entrusted with the role of an anchor. Note also that Eq.~\eqref{bnd1} naturally leads to an upper bound on the neighboring sensor weight, i.e.~$(P_k)_{i,l\neq i}\leq1-\beta_1$, because~$\mc{D}_i(k)$ always includes~$\{i\}$. Also when there is no neighboring anchor, Eq.~\eqref{bnd1} guarantees that sensors do not completely forget their past information by putting a non-zero self-weight on their own previous states. Finally, we point out that the bounds in Eqs.~\eqref{bnd1}--\eqref{bnd3}, are naturally satisfied by LTI dynamics:~$\mb{x}(k+1)=P\mb{x}(k) + B\mb{u}(k)$, with non-negative matrices; a topic well-studied in the context of iterative algorithms,~\cite{tsit_book,plemmons:79}, and multi-agent fusion. 

\section{Infinite product of (sub-) stochastic matrices}\label{IP}
In this section, we study the convergence of 
\begin{equation}\label{eq11}
\lim_{k \rightarrow \infty}{{P}}_{k}{{P}}_{k-1} \ldots {{P}}_0,
\end{equation}
where~${{P}}_{k}$ is the system matrix at time~$k$, as defined in Section~\ref{PF}. Since multiplication with the identity matrix has no effect on the convergence of the sequence, in the rest of the paper we only consider the updates, in which at least one sensor is able to find and exchange information with some neighbors, i.e. ${{P}}_{k} \neq I_n,~\forall k$. We are interested in establishing the stability properties of this infinite product. Studying the joint spectral radius is prone to many challenges as described in Section~\ref{intro}, and we choose the infinity norm to study the convergence conditions. The infinity norm,~$\|M\|_\infty$, of a square matrix,~$M$, is defined as the maximum of the absolute row sums. Clearly, the infinity norm of~${{P}}_k$ is one for all~$k$ since each system matrix has at most one sub-stochastic row. 

To establish a subunit infinity norm, we divide the system matrices into non-overlapping \textit{slices} and show that each slice has an infinity norm strictly less than one; the entire chain of system matrices is covered by these non-overlapping slices. Let one of the slices be denoted by~$M$ with length~$|M|$ and, without loss of generality, index the matrices within~$M$ as
\begin{eqnarray}\label{M}
{{M}} = P_{|{{M}}|}P_{|{{M}}|-1}P_{|{{M}}|-2}\ldots P_{3}P_{2}P_{1}.
\end{eqnarray}

\noindent Using slice notation, we can introduce a new discrete-time index,~$t$, which allows us to study the following
\begin{eqnarray}
\lim_{t \rightarrow \infty}{{M}}_{t}{{M}}_{t-1}\ldots{{M}_{0}},
\end{eqnarray}
instead of Eq.~\eqref{eq11}, note that~$ k > t$. 

We define a system matrix,~$P_k$, as a \textit{success} if it decreases the row sum of some row in~$P_k$, which was stochastic before this successful update. Each success, thus, adds a \emph{new} sub-stochastic row to a slice, and~$n$ such successful updates are required to complete a slice. In this argument, we assume that a row that becomes sub-stochastic remains sub-stochastic, which is not in true in general, after successive multiplication with stochastic or sub-stochastic matrices,~$\left(\ldots P_{k+2}P_{k+1}\right)$. Thus, we will derive the explicit conditions under which the sub-stochasticity of a row is preserved. Before we proceed with our main result we provide the following lemmas:

\begin{lem}\label{lem1}
For the infinity norm to be less than one, each slice has to contain at least one sub-stochastic update.
\end{lem}
\begin{proof}
Since any set of stochastic matrices form a \textit{group} under multiplication \cite{anton2010elementary}, a slice without a sub-stochastic update will be a stochastic matrix whose infinity norm is~$1$.
\end{proof}

We now motivate the slice construction as follows. Partition the rows in an arbitrary~$P_k$ into two distinct sets: set~$\mathcal{I}$ contains all sub-stochastic rows, and the remaining (stochastic) rows form the other set,~$\mathcal{U}$. We initiate each slice with the first success,~$\vert \mathcal{I} \vert=1,~\vert \mathcal{U} \vert=n-1$, and terminate it \textit{after} the~$n$th success,~$\vert \mathcal{I} \vert=n,~\vert \mathcal{U} \vert=0$, when each row becomes sub-stochastic. Between the~$n$th success in the current slice, say~$M_j$, and the first success in the next slice,~$M_{j+1}$, all we can have are stochastic or sub-stochastic matrices that must preserve the sub-stochasticity of each row. See Fig.~\ref{f0} for the slice representation, where the rightmost system matrices (encircled in Fig.~\ref{f0}) of each slice, i.e.~${P}_0,~{P}_{m_0},~\ldots,~{P}_{m_{j-1}} \ldots$, are sub-stochastic. The~$j$th slice length may be defined as
\begin{equation*}
\vert {M}_{j} \vert = m_j - m_{j-1},\qquad m_{-1}=0,
\end{equation*}
and slice lengths are not necessarily equal. 
\begin{figure}[!h]
\centering
\includegraphics[width=2.75in]{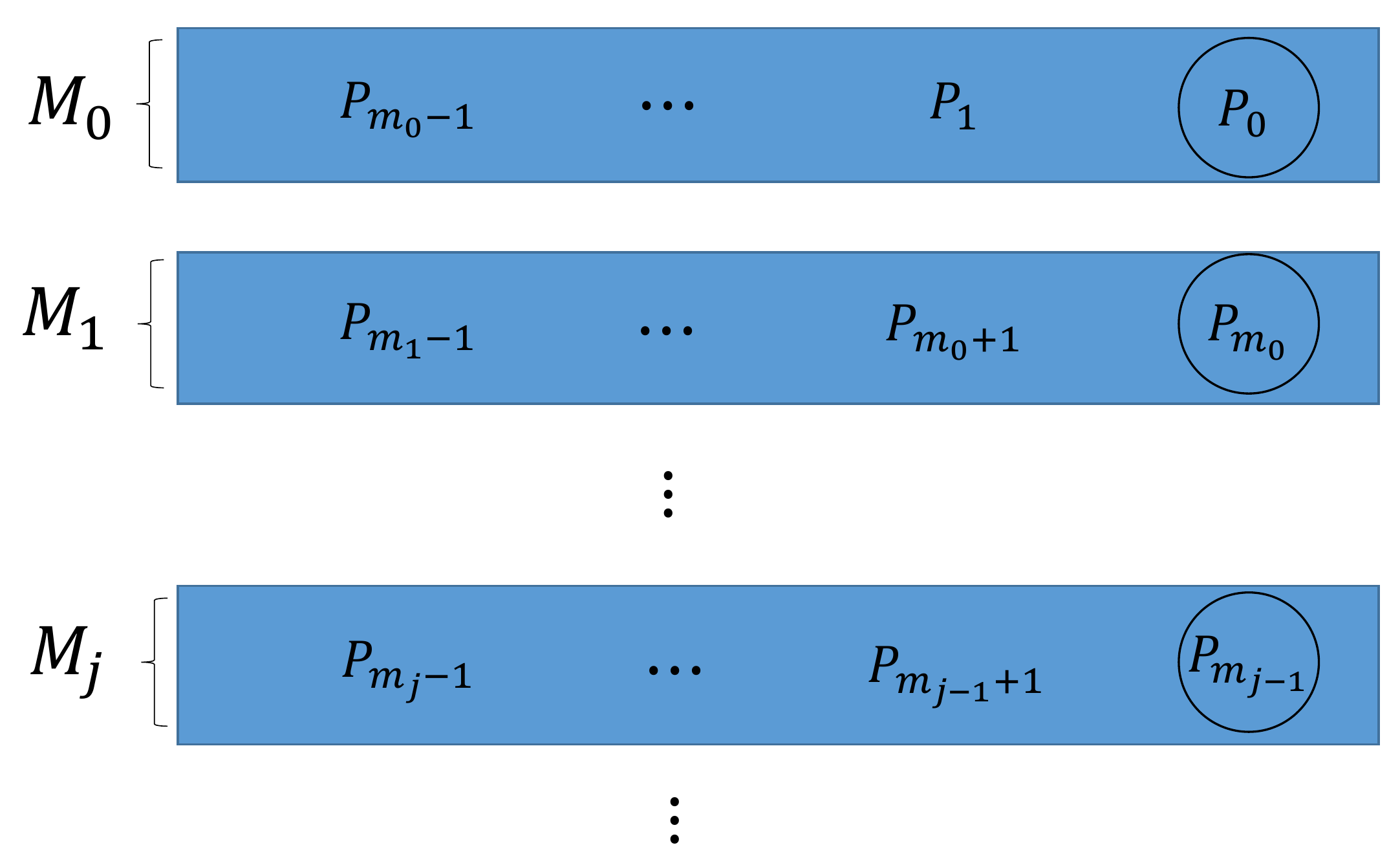}
\caption{Slice representation}
\label{f0}
\end{figure}

In the next lemma, we show how a stochastic row can become sub-stochastic, in a slice,~${{M}}$. We index~$P_k$'s in a slice,~$M$, by~$P_{\vert {M} \vert}, \ldots, P_2, P_1$ to simplify notation, and define the product of all system matrices up to time~$k$ in a slice as
\begin{eqnarray}\label{14}
J_k=P_k P_{k-1} \ldots P_2 P_1, \qquad 0 < k \leq |{{M}}|.
\end{eqnarray}

\begin{lem}\label{lem2}
Suppose the~$i$-th row of~$J_k$ is stochastic at index~$k$ of a given slice,~${{M}}$, and~$P_{k+1}$ is the next system matrix. Row~$i$ in~$J_{k+1}$ can become sub-stochastic by either:\\
(i) a sub-stochastic update at the~$i$-th row of~$P_{k+1}$; or,\\
(ii) a stochastic update at the~$i$-th row of~$P_{k+1}$, such that 
\begin{eqnarray*}
\exists j \in \mathcal{I}_k,\mbox{ with }(P_{k+1})_{ij} \neq 0,
\end{eqnarray*}
where~$\mathcal{I}_k$ is the set of sub-stochastic rows in~$J_k$.
\end{lem} 
\begin{proof}
For the sake of simplicity, let~$P \triangleq P_{k+1}, J \triangleq J_k$, in the following. Updating the~$i$th row at index~$k+1$ leads to
\begin{eqnarray}\label{16}
P_{k+1}\triangleq P=
\left[
\begin{array}{cccccc}
I_{1:i-1}\\
(P)_{i,1} & (P)_{i,2} && \ldots && (P)_{i,n}\\
&&&&& I_{i+1:n}
\end{array}
\right],
\end{eqnarray}
where~$I_n$ is~$n \times n$ identity matrix;~$i$th row after this update is
\begin{eqnarray*}
(PJ)_{i,j} = \sum_{m=1}^n(P)_{i,m}(J)_{m,j},
\end{eqnarray*}
where~$(PJ)_{i,j}$ is the~$(i,j)$th element of~$PJ$, and the~$i$th row sum becomes
\begin{eqnarray}\label{ui}
\sum_{j}(PJ)_{i,j}&=&\sum_{j}\sum_{m=1}^n(P)_{i,m}(J)_{m,j},\\
&=& \sum_{j} \Big((P)_{i,1}(J)_{1,j}+ \ldots + (P)_{i,n}(J)_{n,j} \Big),\nonumber\\
&=&(P)_{i,1}\underbrace{\sum_{j}(J)_{1,j}}_{\leq 1}+ \ldots + (P)_{i,n}\underbrace{\sum_{j}(J)_{n,j}}_{\leq 1}\nonumber.
\end{eqnarray}
Thus, we have
\begin{eqnarray}\label{uii}
\sum_{j}(PJ)_{i,j}\leq (P)_{i,1}+ \ldots + (P)_{i,n}.
\end{eqnarray}
Let us first consider \textit{case (i)} where the~$i$th row of~$P$ is sub-stochastic. From Eq.~\eqref{uii} and Assumption {\bf A2}, we have
\begin{eqnarray}\label{19}
\sum_{j}(PJ)_{i,j}\leq \sum_{j=1}^{n}(P)_{i,j} \leq\beta_2 < 1.
\end{eqnarray}
Therefore, the~$i$th row becomes sub-stochastic after a sub-stochastic update at row~$i$. 

We now consider \textit{case~(ii)} where the~$i$th row of~$P$ is stochastic, i.e.~$\sum_{m=1}^{n}{(P)}_{i,m} = 1$. In this case,~$\sum_{j}(PJ)_{i,j}$ is a linear-convex combination of the row sums of~$J$, which is strictly less than one, if and only if~$J$ has at least one sub-stochastic row, say~${m}^{\prime}$, such that~${(P)}_{i,{m}^{\prime}} \neq 0$, i.e.
\begin{eqnarray}\label{20}
\sum_{j}(PJ)_{i,j}=\underbrace{(P)_{i,{m}^{\prime}}}_{\neq 0}\underbrace{(J)_{{m}^{\prime},j}}_{<1}+\sum_{m\neq{m}^{\prime}}(P)_{i,m}\underbrace{(J)_{m,j}}_{\leq1}.
\end{eqnarray}
So~$\sum_{j}(PJ)_{i,j} <1$ and the lemma follows.
\end{proof}

In the next lemma, we show that sub-stochasticity is preserved for each sub-stochastic row within a slice.
\begin{lem}\label{lem3}
With Assumptions {\bf{A0-A2}}, a sub-stochastic row, say~$i$, remains sub-stochastic throughout a slice.
\end{lem} 
\begin{proof}
We use the notation of Lemma~\ref{lem2} on~$J$,~$P$, and Eq.~\eqref{16}, and rewrite Eq.~\eqref{ui} as
\begin{eqnarray}
\sum_{j}(PJ)_{i,j}&=&\sum_{j}\sum_{m=1}^n(P)_{i,m}(J)_{m,j},\nonumber\\
&=& \sum_{m \in \mathcal{I}} \Bigg((P)_{i,m}\sum_{j}(J)_{m,j} \Bigg) \nonumber\\
&+& \sum_{m \in \mathcal{U}} \Bigg((P)_{i,m}{\sum_{j}(J)_{m,j}} \Bigg).
\end{eqnarray}
Let us consider the general case after the first success, where there exist~$r \geq 1$ sub-stochastic rows in~$J$, i.e.~$\vert \mathcal{I} \vert=r$, and~$\vert \mathcal{U} \vert=n-r$. Without loss of generality, suppose the~$r$ sub-stochastic rows of~$J$ lie in the first~$r$ rows. We need to show that if the~$i$th row in~$J$ is sub-stochastic, i.e.~$i\leq r$, it remains sub-stochastic after a multiplication by either a stochastic or a sub-stochastic system matrix,~$P$. Rewrite the~$i$th row sum as
\begin{eqnarray}\label{27}
\sum_{j}(PJ)_{i,j}=
(P)_{i,1}\sum_{j}(J)_{1,j}+ \ldots + (P)_{i,r}\sum_{j}(J)_{r,j}\nonumber\\
+ (P)_{i,r+1}\underbrace{\sum_{j}(J)_{r+1,j}}_{=1}+ \ldots + (P)_{i,n}\underbrace{\sum_{j}(J)_{n,j}}_{=1}\nonumber.
\end{eqnarray}
Thus,
\begin{eqnarray}
\sum_{j}(PJ)_{i,j}&=&(P)_{i,1}\sum_{j}(J)_{1,j}+ \ldots + (P)_{i,r}\sum_{j}(J)_{r,j}\nonumber\\
&+&(P)_{i,r+1}+\ldots+(P)_{i,n},
\end{eqnarray}
where we used the fact that in~$J$, any row~$\in \mathcal{U}$ is stochastic. 

\emph{Let us first consider the~$i$th row of~$P$ to be stochastic}:
\begin{eqnarray*}
(P)_{i,1}+\ldots+(P)_{i,r}+(P)_{i,r+1}+\ldots+(P)_{i,n} = 1.
\end{eqnarray*}
Thus,
\begin{eqnarray*}
(P)_{i,r+1}+\ldots+(P)_{i,n} = 1 - \Big((P)_{i,1}+\ldots+(P)_{i,r}\Big).
\end{eqnarray*}
Therefore, from Eq.~\eqref{27} we can write
\begin{eqnarray}
\sum_{j}(PJ)_{i,j}&=&
(P)_{i,1}\sum_{j}(J)_{1,j}+ \ldots + (P)_{i,r}\sum_{j}(J)_{r,j}\nonumber\\
&+& 1- \Big((P)_{i,1}+\ldots+(P)_{i,r}\Big).
\end{eqnarray}
Finally,
\begin{eqnarray}\label{29}
0 \leq \sum_{j}(PJ)_{i,j} = & 1 +&  (P)_{i,1}\Bigg(\sum_{j}(J)_{1,j} -1\Bigg)\nonumber\\
&\vdots& \nonumber\\
&+& (P)_{i,r}\Bigg(\sum_{j}(J)_{r,j} -1\Bigg),\\
\leq & 1 +&  (P)_{i,i}\Bigg(\sum_{j}(J)_{i,j} -1\Bigg),\label{25}
\end{eqnarray}
because the first~$r$ rows in~$J$ are sub-stochastic leading to~$\sum_j(J)_{m,j}-1<0$, for any~$m=1,\ldots,i,\ldots r$, and since the~$i$th row in~$P$ is stochastic, by Assumption {\bf{A1}} we have
\[
0<\beta_1\leq (P)_{i,i}.
\]
Note that in Eq.~\eqref{29} the only way to lose sub-stochasticity is to have~$(P)_{i,m}=0$ for all~$m\leq r$. However, sub-stochasticity can be preserved by putting a non-zero weight on any row in~$\mc{I}$. Since this knowledge in not available in general, a sufficient condition to ensure this is~$0<\beta_1\leq(P)_{i,i}$. Thus, the~$i$th row sum remains strictly less than one (and greater than zero) after any stochastic update at the~$i$th row as long as Assumption {\bf{A1}} is satisfied. Note that the lower bound on the~$j$th row sum stems from the non-negativity of system matrices. 

\emph{Now consider the~$i$th row of~$P$ to be sub-stochastic}. From {\bf{A2}}, we have
\begin{eqnarray*}
(P)_{i,1}+\ldots+(P)_{i,r}+(P)_{i,r+1}+\ldots+(P)_{i,n} \leq \beta_2 <1.
\end{eqnarray*}
Therefore,
\begin{eqnarray*}
(P)_{i,r+1}+\ldots+(P)_{i,n} \leq \beta_2 - \Big((P)_{i,1}+\ldots+(P)_{i,r}\Big).
\end{eqnarray*}
Thus, from Eq.~\eqref{27} we can write
\begin{eqnarray}
\sum_{j}(PJ)_{i,j}&\leq&
(P)_{i,1}\sum_{j}(J)_{1,j}+ \ldots + (P)_{i,r}\sum_{j}(J)_{r,j}\nonumber\\
&+& \beta_2- \Big((P)_{i,1}+\ldots+(P)_{i,r}\Big).
\end{eqnarray}
Finally,
\begin{eqnarray}\label{29-}
0 \leq \sum_{j}(PJ)_{i,j}\leq \beta_2 &+&  (P)_{i,1}\Bigg(\sum_{j}(J)_{1,j} -1\Bigg)\nonumber\\
&\vdots& \nonumber\\
&+& (P)_{i,r}\Bigg(\sum_{j}(J)_{r,j} -1\Bigg),\nonumber\\
\leq &\beta_2 & <1,
\end{eqnarray}
where again we used the fact that~$\sum_j(J)_{m,j}-1<0,m=1,\ldots,i,\ldots r$. Eq.~\eqref{29-} shows that in case of a sub-stochastic~$i$th row in~$P_{k+1}$, this row remains sub-stochastic in~$J_{k+1}$, as long as Assumption {\bf{A2}} is satisfied and the conditions on individual weights are not required. Note the strict inequality, i.e. if~$(P)_{i,m}\neq0$ for any~$m=1,\ldots,r$, then
\begin{eqnarray*}
\sum_{j}(PJ)_{i,j} < \beta_2.
\end{eqnarray*}
This lemma establishes that under the Assumptions {\bf A0-A2}, sub-stochasticity is always preserved.
\end{proof} 

The results so far describe the behavior of the sub-stochastic rows in the slices explicitly derived under the regularity conditions in Assumptions {\bf A0-A2}. The next results characterize the infinity norm bound on the slices. To this end, let us define~${\beta}_4(j)$, as the \emph{maximum row sum over the sub-stochastic rows} of the product of all system matrices before~$P_j$ in the~$M$th slice. Mathematically,
\begin{eqnarray}\label{beta4}
{\beta}_4(j)= \max\limits_{m \in \mathcal{I}_{j-1}} \{{v}_m\},
\end{eqnarray}
where~$v_m$ is the~$m$th element of the following column vector
\begin{eqnarray*}
\mathbf{v}_{j-1}= (J_{j-1}){\textbf{1}_n}=(P_{j-1} \ldots P_{3}P_{2}P_{1}){\textbf{1}_n},
\end{eqnarray*}
and~${\textbf{1}}_n$ is the column vector of~$n$ ones.

It can be inferred from our discussion so far that a sub-stochastic update at row~$i$ is sufficient but not necessary for the~$i$th row to be sub-stochastic. In the following lemma, we consider the case where no sub-stochastic update occurs at row~$i$ throughout a slice, and provide an upper bound for the~$i$th row sum at the end of a slice.
\begin{lem}\label{lm4}
Assume there is no sub-stochastic update at the~$i$-th row within a given slice,~$M$. The~$i$-th row sum of this slice is upper bounded by
\begin{eqnarray}
1+{\beta_1}^{\vert {M}\vert - {h}_{i} +1}({\beta}_4({h}_{i})-1), 
\end{eqnarray}
where the first success at row~$i$ occurs in the~${h}_{i}$-th update of this slice.
\end{lem}

\begin{proof}
Eq.~\eqref{29} expresses the~$i$th row sum after a stochastic update at row~$i$. Clearly, before the first success, at index~$h_i$,
\begin{eqnarray*}
\sum_{j}(J_k)_{i,j}=1,\qquad \forall k < {h}_{i}.
\end{eqnarray*}
In order to find the maximum possible row sum for the~$i$th row at the end of a slice, we should find a scenario, which maximizes the row sum after the first success at index~$h_i$ and keeps maximizing it at each subsequent update. Let us consider Eq.~\eqref{29} after the first success at index~$h_i$. Since no sub-stochastic update is allowed at row~$i$ from the lemma's statement, the first success occurs via a stochastic update at the~$i$th row, and Assumption~{\bf A1} is applicable. Since any non-zero~$(P)_{i,m \in \mathcal{I}}$ decreases the row sum, the minimum number of such weights maximizes the right hand side (RHS) of Eq.~\eqref{29}. Suppose~$(P_{{h}_{i}})_{i,{r}^{\prime}}$ is the only non-zero among all~$(P_{{h}_{i}})_{i,j \in \mathcal{I}}$'s. In this case, Eq.~\eqref{29} reduces to the following
\begin{eqnarray}\label{30}
\sum_{j}(P_{{h}_{i}} J_{{{h}_{i}-1}})_{i,j}= 1 -  (P_{{h}_{i}})_{i,{r}^{\prime}}\Bigg(1-\sum_{j}(J_{{{h}_{i}-1}})_{{r}^{\prime},j}\Bigg),
\end{eqnarray}
in which~${r}^{\prime} \in \mathcal{I}_{{h}_{i}-1}$. Also note that~${r}^{\prime} \neq i$, since~$i$ is stochastic before the time instant,~$h_i$. In order to maximize the RHS of Eq.~\eqref{30},~$(P_{{h}_{i}})_{i,{r}^{\prime}}$ should be minimized, and~$\sum_{j}{(J_{{{h}_{i}-1}})}_{{r}^{\prime},j}$ should be maximized. From Eq.~\eqref{beta4}, the maximum value of~$\sum_{j}{(J_{{{h}_{i}-1}})}_{{r}^{\prime},j}$ before the first success is~${\beta}_4({h}_{i})$. Thus, after the~${h}_{i}$th update, where row~$i$ becomes sub-stochastic for the first time, we can write
\begin{eqnarray}\label{31}
\sum_{j}(P_{{h}_{i}} J_{{{h}_{i}-1}})_{i,j}\leq 1 - \beta_1 \Bigg(1-{\beta}_4({h}_{i})\Bigg).
\end{eqnarray}
After this update,~$J_{{h}_{i}}=P_{{h}_{i}}\ldots P_2 P_1$, and
\begin{eqnarray}
{\beta}_4({h}_{i}) \leq {\beta}_4({h}_{i}+1) \leq 1 - \beta_1 \Bigg(1-{\beta}_4({h}_{i})\Bigg),
\end{eqnarray}
where~${\beta}_4({h}_{i})$ is the~$r^{\prime}$th row sum in~$J_{{h_i}-1}$, and~${\beta}_4({h}_{i}+1)$ is the~$i$th row sum in~$J_{h_i}$.
Under this scenario, after the first success, the~$i$th row has the maximum row sum over all rows of~$J_{{h}_{i}}$, and in order to increase this row sum at the next update, the~$i$th row has to update only with itself. Note that after the success at index~$h_i$, row~$i$ becomes sub-stochastic, and~${r}^{\prime}=i$, for any subsequent update until the end of a slice. After the next update,~$P_{{h}_{i}+1}$, using the same argument we can write
\begin{eqnarray}\label{32}
\sum_{j}(P_{{h}_{i}+1} J_{{h}_{i}})_{i,j} & = & 1 -  (P_{{h}_{i}+1})_{i,i}\Bigg(1-\sum_{j}(J_{{h}_{i}})_{i,j}\Bigg),\nonumber\\
&\leq& 1 - \beta_1 \Bigg(1- (1 - \beta_1 (1-{\beta}_4({h}_{i}))\Bigg)\nonumber\\
&=& 1 - {\beta_1}^{2} \Bigg(1-{\beta}_4({h}_{i})\Bigg).
\end{eqnarray}
If row~$i$ keeps updating with itself, at the end of slice, we have after~${\vert {M}\vert - {h}_{i}}$ number of such updates
\begin{eqnarray}\label{33}
\sum_{j}({P}_{\vert {M}\vert} J_{\vert {M}\vert -1 })_{i,j} \leq 1+{\beta_1}^{\vert {M}\vert - {h}_{i}+1}({\beta}_4({h}_{i}) -1),
\end{eqnarray}
and the lemma follows.
\end{proof}

In the following lemma, we consider the general case where sub-stochastic updates are also allowed at row~$i$ and provide an upper bound for the~$i$th row sum at the end of a slice.
\begin{lem}\label{lm5}
Assume there is at least one sub-stochastic update at the~$i$-th row within a given slice,~$M$. The~$i$-th row sum of this slice is upper bounded by
\begin{eqnarray}
1+{\beta_1}^{\vert {M}\vert - {g}_{i}}({\beta}_2 -1),
\end{eqnarray}
where the last sub-stochastic update at row~$i$ occurs in the~${g}_{i}$-th update of a slice.
\end{lem}
\begin{proof}
As shown in Eq.~\eqref{29-} and by Assumption {\bf A2}, any sub-stochastic update at row~$i$ imposes the upper bound of~$\beta_2$ on the~$i$th row sum.
Thus, after the \textit{last} sub-stochastic update at row~$i$ we have
\begin{eqnarray*}
\sum_{j}(P_{{g}_{i}} J_{{g}_{i}-1})_{i,j}\leq \beta_2 < 1.
\end{eqnarray*}
After~$P_{{g}_{i}}$, there is no sub-stochastic update, and by Assumption {\bf{A1}}, the~$i$th self-weight will be non-zero until the end of the slice. Following the same argument as in Lemma~\ref{lm4}, the upper bound on the~$i$th row sum is maximized after each update if the~$i$th row does not update with any  sub-stochastic row other than itself. For any update after the last success, Eq.~\eqref{bnd1} holds and we have
\begin{eqnarray}
{\mc{N}_{i}(k)}\cap\mathcal{I}_k = \emptyset,\qquad \forall k > {g}_{i}.
\end{eqnarray} 
After the~$P_{{g}_{i}+1}$th update we have
\begin{eqnarray}\label{34}
\sum_{j}(P_{{g}_{i}+1} J_{{g}_{i}})_{i,j}\leq 1 - \beta_1 \Bigg(1-{\beta}_2\Bigg),
\end{eqnarray}
and at the end of a slice, we have
\begin{eqnarray}\label{35}
\sum_{j}({P}_{\vert {M}\vert} J_{\vert {M}\vert -1 })_{i,j} \leq 1+{\beta_1}^{\vert {M}\vert - {g}_{i}}({\beta}_2 -1),
\end{eqnarray}
and the lemma follows.
\end{proof}

In the previous two lemmas, we provide an upper bound for each row sum for two cases: when all updates are stochastic \emph{and} when sub-stochastic updates are also allowed. The following lemma combines these bounds and relate them to the infinity norm bound of a slice.

\begin{lem}
For a given slice,~${{{M}}}$, 
\begin{eqnarray}\label{bound1}
\|M\|_\infty\leq\max\limits_{i}\{1+{\beta_1}^{\vert {M}\vert - l_i}({\beta} -1)\},
\end{eqnarray}
where
\begin{eqnarray*}
l_i &=& h_i - 1,~\beta=\beta_4(h_i),~\mbox{stochastic updates at row i},\\
l_i &=& g_i,~\beta=\beta_2,~\mbox{(sub-) stochastic updates at row i}.
\end{eqnarray*}
\end{lem}

\noindent The next lemma studies the worst case scenario for the infinity norm of a slice, which provides an upper bound for Eq.~\eqref{bound1}.

\begin{lem}\label{lm7}
With assumptions {\bf{A0-A2}}, for the~$j$th slice we have
\begin{eqnarray}\label{41}
{{\Vert M_j \Vert}_{\infty}} \leq 1-{\alpha}_j <1, \qquad j\geq0,
\end{eqnarray}
where 
\begin{eqnarray}\label{42}
{\alpha}_j = f({{\vert M_j \vert}},\beta_1,\beta_2)={\beta_1}^{{{\vert M_j \vert}} - 1}({1-\beta}_2).
\end{eqnarray}
\end{lem}
\begin{proof}
In order to find the maximum upper bound on the infinity norm of a slice, we  consider a \textit{worst case scenario}, in which a row sum incurs the largest increase throughout the slice. To do so, we examine the maximum possible upper bound on the~$i$th row sum for the two cases discussed in Lemmas~\ref{lm4} and~\ref{lm5} separately. 

Consider no sub-stochastic update at the~$i$th row. We should find a scenario that maximizes the RHS of Eq.~\eqref{33}. In addition, we need to make sure that such scenario is practical, i.e. all other rows become sub-stochastic before a slice is terminated. Since there are no sub-stochastic updates at row~$i$, a slice can not be initiated by an update in row~$i$, i.e.~$h_i \geq 2$. At the initiation of a slice, one row other than~$i$, becomes sub-stochastic, and the upper bound imposed on this row is~$\beta_2$ by Assumption~{\bf{A2}}, hence~${\beta}_4(h_i)={\beta}_4(2)=\beta_2$. Therefore, following the discussion in Lemma~\ref{lm4},
\begin{eqnarray}\label{43}
1+{\beta_1}^{\vert {M_j}\vert - 1}({\beta}_2-1),
\end{eqnarray}
provides the largest upper bound on the~$i$th row sum of $M_j$. Note that this bound is feasible if we consider the following scenario. After row~$i$ becomes sub-stochastic at~$h_i=2$ we let next~$n-2$ updates for the other stochastic rows to become sub-stochastic, each updating \textit{only} with the sub-stochastic row with the largest row sum. Thus the largest row sum keeps increasing in the same manner as discussed in Lemma~\ref{lm4} within the next~$n-2$ updates. At~$n+1$th update, row~$i$ again updates with a row, which has the maximum row sum in~$J_n$, and keeps updating by itself until the slice is terminated. The aforementioned scenario is equivalent to the one where the first success at row~$i$ occurs at~$h_i=n$, and all other rows become sub-stochastic within the first~$n-1$ updates, and 
\begin{equation}
\beta_4(h_i=n)= 1 + {\beta_1}^{n-2}(\beta_2-1).
\end{equation}

Now consider sub-stochastic updates at row~$i$. The RHS of Eq.~\eqref{35} is maximized if~$g_i$ is minimized. In this case, the minimum value for~$g_i$ is one, which corresponds to a scenario where a sub-stochastic update at row~$i$ initiates a slice and no other sub-stochastic update occurs at this row. Using the same argument as before, all other rows become sub-stochastic within the next~$n-1$ updates and the largest upper bound on the~$i$th row in this case is the same as the one given in Eq.~\eqref{43}.
\end{proof}
Finally, note that for a given slice,~$M$, 
\begin{eqnarray}
\|M\|_\infty\leq 1+ {\beta_1}^{\vert {{{M}}} \vert-1}({\beta_2}-1)
\end{eqnarray}
is the \textit{largest} upper bound on the infinity norm of a slice.

\section{Stability of discrete-time  systems}\label{stability}
In this section, we study the stability of discrete-time, LTV dynamics with (sub-) stochastic system matrices. We start with the following definitions:
\begin{definition}
The system represented in Eq.~\eqref{eq2} is \textit{asymptotically stable} (or convergent) if for any~${\bf{x}}(0)$,
\begin{equation}
\lim_{k \rightarrow \infty}{\bf{x}}(k) \nonumber
\end{equation}
is bounded and convergent.
\end{definition}

\begin{definition}
The system represented in Eq.~\eqref{eq2} is \textit{absolutely asymptotically stable} (or zero-convergent) if for any~${\bf{x}}(0)$, 
\begin{equation}
\lim_{k \rightarrow \infty}{\bf{x}}(k)=0.\nonumber
\end{equation}
\end{definition}

Recall that we are interested in the asymptotic stability of Eq.~\eqref{eq2}, such that the steady-state forgets the initial conditions and is a function of inputs. A sufficient condition towards this aim is the absolutely asymptotic stability of the following:
\begin{eqnarray}\label{eqq15}
{\bf{x}}(k+1)&=&{{P}_k}{\bf{x}}(k), \qquad k\geq0,\\
&=&{P}_{k} {P}_{k-1} \ldots {{P}_0}{\bf{x}}(0),\nonumber
\end{eqnarray}
for any~${\bf{x}}(0)$, which is equivalent to having 
\begin{equation}
\lim_{k \rightarrow \infty} {P}_{k} {P}_{k-1} \ldots {{P}_0}= \mb{0}_{n\times n},
\end{equation}
where the subscript below $\mb{0}$ denote its dimensions. As depicted in Fig.~\ref{f0}, we can take advantage of the slice representation and study the following dynamics:
\begin{eqnarray}\label{eqq26}\label{50}
{\bf{y}}(t+1)={{M}}_{t}{\bf{y}}(t), \qquad t \geq 0,
\end{eqnarray}
instead of Eq.~\eqref{eqq15}, where
\begin{eqnarray*}
{\bf{y}}(0)&=&{\bf{x}}(0),\\
{\bf{y}}(t)&=& {\bf{x}}\left(t_1\right), \qquad t\geq1, t_1=\sum\limits_{i=1}^{t}{\vert M_i \vert}.
\end{eqnarray*}
Thus, for absolutely asymptotic stability of Eq.~\eqref{eqq26}, for any~${\bf{y}}(0)$, we require
\begin{eqnarray}\label{49}
\lim_{t \rightarrow \infty}{\bf{y}}(t+1)&=&\lim_{t \rightarrow \infty} {{M}}_{t}{\bf{y}}(t),\\
&=& \lim_{t \rightarrow \infty}{{M}}_{t}{{M}}_{t-1} \ldots {{M}}_{0}{\bf{y}}(0),\nonumber\\
&=& \mb{0}_n.\nonumber
\end{eqnarray}
We provide our main result in the following theorem.
\begin{thm}\label{thm0}
With assumption {\bf{A0-A2}}, the LTV system in Eq.~\eqref{eqq26} is absolutely asymptotically stable if either one of the following is true:
\begin{enumerate}[(i)]
\item Each slices has a bounded length, i.e.
\begin{eqnarray}\label{51}
\vert M_j \vert \leq N < {\infty} , \qquad \forall j,~N \in \mathbb{N};
\end{eqnarray}

\item There exist a set,~$J_1$, consisting of an infinite number of slices such that 
\begin{eqnarray}
\vert M_{j} \vert \leq N_1 < {\infty}, \qquad \forall M_j\in J_1,\\
\vert M_{j} \vert < {\infty}, \qquad \forall M_j \notin J_1;
\end{eqnarray}

\item There exists a set,~$J_2$, of slices such that
\begin{eqnarray*}
\exists M_j \in J_2:~~\vert M_j \vert \leq \frac{1}{\ln\left({\beta_1}\right)}\ln\left(\frac{1 - e^{(-\gamma_2i^{-\gamma_1})}}{1-\beta_2}\right)+1,
\end{eqnarray*}
for every~$i \in \mathbb{N}$, and~$|M_j|<\infty,j\notin J_2$.
\end{enumerate}
\end{thm}
\begin{proof}
Using the sub-multiplicative norm property, Eq.~\eqref{50} leads to
\begin{align}\label{55}
{\Vert{\bf{y}}(t+1)\Vert}_\infty &\leq {\Vert {{M}}_{t} \Vert}_{\infty}  \ldots {\Vert {{M}}_{0} \Vert}_{\infty} {\Vert {\bf{y}}(0) \Vert}_{\infty}.
\end{align}

\emph{Case (i)}: From Eqs.~\eqref{41},~\eqref{42} and~\eqref{51}, we have
\begin{eqnarray}
{{\Vert M_j \Vert}_{\infty}} \leq \delta < 1,\qquad\forall j,
\end{eqnarray}
where $\delta=1+{\beta_1}^{{N} - 1}({\beta}_2-1) <1$, and this case follows. 

\textit{Case (ii)}: We first note that the infinity norm of each slice has a trivial upper bound of~$1$. From Eq.~\eqref{55}, we have
\begin{align}
\lim_{t \rightarrow \infty}{\Vert {\bf{y}}(t+1) \Vert}_{\infty} &\leq \lim_{t \rightarrow \infty}\prod\limits_{j \in J_1} {\Vert{M}_{j}\Vert}_\infty \prod\limits_{j \notin J_1}{\Vert{M}_{j}\Vert}_\infty{\Vert {\bf{y}}(0) \Vert}_{\infty},\nonumber\\
&\leq \lim_{t \rightarrow \infty} \prod\limits_{j \in J_1} {\Vert{M}_{j}\Vert}{\Vert {\bf{y}}(0) \Vert}_{\infty}.
\end{align}
Similar to case (i), this case follows by defining
\begin{eqnarray*}
\|M_j\|_\infty\leq \delta_{1}=1+{\beta_1}^{{N_1} - 1}({\beta}_2-1) <1.,
\end{eqnarray*}

\textit{Case (iii)}: With~$\alpha_j$ in Eq.~\eqref{42}, Eq.~\eqref{55} leads to
\begin{eqnarray}\label{62}
\lim_{t \rightarrow \infty}{\Vert {\bf{y}}(t+1) \Vert}_{\infty} \leq \lim_{t \rightarrow \infty} \prod_{j=0}^t(1-\alpha_j) {\Vert {\bf{y}}(0) \Vert}_{\infty}.
\end{eqnarray}
Consider the asymptotic convergence of the infinite product of a sequence~$1-\alpha_j$ to~$0$. We have
\begin{eqnarray}\label{lne1}
\lim_{t \rightarrow \infty}\prod_{j=1}^t(1-\alpha_j)=0,\mbox{ or } \lim_{t \rightarrow \infty}\sum_{j=1}^t(-\ln(1-\alpha_j))=\infty.
\end{eqnarray}
Now note that
\[
\sum_{i=1}^{\infty}\gamma_2i^{-\gamma_1} = \infty, \qquad \mbox{for }0\leq\gamma_1\leq 1,0<\gamma_2,
\]
because~$\frac{1}{i^{\gamma_1}}$ sums to infinity for all values of~$\gamma_1$ in~$[0,1]$, and multiplying by a positive number,~$\gamma_2$, does not change the infinite sum. It can be verified that Eq.~\eqref{lne1} holds when 
\begin{eqnarray*}
-\ln(1-\alpha_j) &\geq& \gamma_2i^{-\gamma_1},
\end{eqnarray*}
subsequently resulting into 
\begin{eqnarray*}
1-\alpha_j &\leq& e^{(-\gamma_2i^{-\gamma_1})},
\end{eqnarray*}
for some~$\gamma_1\in[0,1]$, and~$0<\gamma_2$.
Therefore if for any~$i \in \mathbb{N}$, there exist a slice,~$M_j$, in the set,~$J_2$, such that
\begin{eqnarray}\label{64-0}
\|M_j\|_\infty\leq 1-\alpha_j \leq e^{(-\gamma_2i^{-\gamma_1})},
\end{eqnarray}
we get
\begin{eqnarray}\label{57}
\lim_{t \rightarrow \infty}\prod_{j=0}^t(1-\alpha_j) = \underbrace{\prod_{j\in J_2}(1-\alpha_j)}_{=0} \prod_{j\notin J_2}(1-\alpha_j)=0,
\end{eqnarray}
and absolutely asymptotic stability follows. By substituting~$\alpha_j$ from Eq.~\eqref{42} in the left hand side of Eq.~\eqref{64-0}, we get
\begin{eqnarray}\nonumber
1 - {\beta_1}^{{{\vert M_j \vert}} - 1}({1-\beta}_2) &\leq&  e^{(-\gamma_2i^{-\gamma_1})},
\end{eqnarray}
which leads to 
\begin{eqnarray}\label{65}
\ln\left(\frac{1 - e^{(-\gamma_2i^{-\gamma_1})}}{1-\beta_2}\right) &\leq&  (\vert M_j \vert - 1)\ln\beta_1.\label{58}
\end{eqnarray}
Since $\beta_1< 1$, $\ln\beta_1$ is negative and dividing both sides of~Eq.~\eqref{58} by a negative number changes the inequality, i.e.
\begin{eqnarray}\label{c3}
\vert M_j \vert \leq \frac{1}{\ln\left({\beta_1}\right)}\ln\left(\frac{1 - e^{(-\gamma_2i^{-\gamma_1})}}{1-\beta_2}\right)+1. 
\end{eqnarray}
Now note that the first~$\ln$ is negative; for the bound to remain meaningful, the second~$\ln$ must also be negative that requires 
\begin{align*}
1 - e^{(-\gamma_2i^{-\gamma_1})} &< 1-\beta_2,\\ \mbox{ or, } \beta_2 &< e^{(-\gamma_2i^{-\gamma_1})}.
\end{align*}
It can be verified that the above inequality is true for any value of $\beta_2\in[0,1)$ by choosing an appropriate $0<\gamma_2$. 

To conclude, we note that if the slices are such that there exists a slice with length following Eq.~\eqref{c3} for every~$i\in\mbb{N}$, not necessarily in any order, the infinite product of such slices goes to a zero matrix. Finally, from Eq.~\eqref{62}
\begin{eqnarray}
\lim_{t \rightarrow \infty}{\Vert {\bf{y}}(t+1)\Vert}_{\infty}=0,
\end{eqnarray}
which completes the proof in this case.
\end{proof}

In the following, we shed some light on case (iii) and Eq.~\eqref{c3}. First, note that Eq.~\eqref{c3} does not require the slice indices to be~$i$. In other words, the slice lengths are not growing as~$i$ increases and slices satisfying Eq.~\eqref{c3} may appear in any order. For the next argument, note that the RHS of Eq.~\eqref{c3} goes to~$+\infty$ as $i\ra\infty$; because $e^{(-\gamma_2 i^{-\gamma_1})}$ goes to~$1$. A longer slice length can be related to a slow information propagation in the network. Eq.~\eqref{c3} further shows that LTV stability does not require bounded slice lengths (as in cases (i) and (ii)); the slice lengths can be unbounded as long as a well-behaved sub-sequence of slices exist (in any order) whose lengths do not increase faster than the upper bound in Eq.~\eqref{c3}. 

Next note that $\gamma_1=1$ is a valid choice, which corresponds to the fastest growing exponential,~$e^{(-\gamma_2 i^{-1})}$, whose infinite product is $0$. 
This means that only a sub-sequence of slices need to behave such that their behavior is not worse that~${e^{-\gamma_2i^{-1}}}$, in any order. We may write this requirement as 
\begin{eqnarray*}
\mbb{P}\left(M_j \mbox{ exists for some~$j$ such that }\|M_j\|_\infty\leq {e^{-\gamma_2i^{-1}}}\right)=1,
\end{eqnarray*}
$\forall i\geq1$ and~$0<\gamma_2$, where $\mbb{P}$ denotes the probability of the corresponding event. On the other hand, by choosing $\gamma_1=0$ the upper bound on the slice length in \textit{case (iii)} becomes a constant. Hence, the first two cases are in fact special cases of this bound if we set $N$ and $N_1$ as
\begin{eqnarray*}
\frac{1}{\ln\left({\beta_1}\right)}\ln\left(\frac{1 - e^{(-\gamma_2)}}{1-\beta_2}\right)+1.
\end{eqnarray*}

\section{Distributed Dynamic Fusion}\label{app}
We now show the relevance of the results in Sections~\ref{IP} and~\ref{stability} to Distributed Dynamic Fusion (DDF) that we briefly introduced in Sections~\ref{intro} and~\ref{PF}. In order to explain the DDF, let us first consider LTI fusion of the form:~$\mb{x}_{k+1} = P\mb{x}_k+B\mb{u}_k$, where~$\mb{x}_k\in\mbb{R}^n$ is the vector of~$n$ sensor states and~$\mb{u}_k\in\mbb{R}^s$ is the vector~$s$ anchor states. The matrix~$P$ collects the sensor-to-sensor coefficients while the matrix~$B$ collects the sensor-to-anchor coefficients. It is clear that if the spectral radius of~$P$ is subunit, the sensor states,~$\mb{x}_k$, forget the initial states,~$\mb{x}_0$, and converge to the convolution between the system's impulse response~$(I-P)^{-1}B$ and the anchor states,~$\mb{u}_k$. When the system matrices are designed such that the concatenated matrix,~$[P~B]$, is row-stochastic, then a subunit spectral radius, $\rho(P)<1$, can be guaranteed if each sensor has a path from at least one anchor. With these conditions, the constant system matrix, $P_k=P,B_k=B,$ at each~$k$, ensure that the information travels from the anchors to each sensor infinitely often and in an exact same fashion at each~$k$. 

In the context of DDF, the system matrices,~$P_k$ and~$B_k$, are a function of the network configuration, and the LTI information flow cannot be guaranteed for any~$P_k,B_k$. In fact, there can be situations when every (mobile) sensor has no neighbors resulting into~$\mb{x}_{k+1}=\mb{x}_k$, i.e.~$P_k=I_n$ and $B_k=\mb{0}_{n\times s}$. The construct of slices ensures that the aforementioned information flow (each sensor having a path from at least one anchor, possibly in arbitrarily different ways) is guaranteed over each slice. In this sense, the \emph{success} regarded (earlier in Section~\ref{IP}) as having a sub-stochastic row, say~$i$, in some arbitrary~$P_k$ in an arbitrary slice,~$M_j$, is equivalent to saying that sensor~$i$ is now \emph{informed}, i.e. sensor~$i$'s current state is now influenced by the anchor(s) in the~$j$th slice. Having~$n$ such (distinct) successes means that in the~$j$th slice, each sensor is now informed. Having infinite such slices means that each sensor becomes informed infinitely often; compare this with the LTI case when all of the~$n$ sensors become informed at each~$k$ and this process repeats infinitely often over~$k$.

The results in Sections~\ref{IP} and~\ref{stability} can also be cast in the context of the DDF discussion above. Lemma~\ref{lem1} states that for each sensor to become informed in every slice, one sensor has to directly receive information from an anchor. Lemma~\ref{lem2} shows how an \emph{uninformed} sensor, say~$i$, may become informed in each slice: either via an anchor, i.e. a sub-stochastic update at row~$i$, or via an (already) informed sensor, i.e. a stochastic update at row~$i$ but with a non-zero weight on any informed sensor. Subsequently, Lemma~\ref{lem3} shows that a sufficient condition for an informed sensor to remain informed in each slice is to assign a non-zero self-weight, i.e. Assumption~{\bf A1}; this makes sense as an informed sensor may become uninformed by updating only with uninformed sensors in its neighborhood.

Lemmas~\ref{lm4}--\ref{lm7} further quantify the rate at which each slice is completed, i.e. the rate at which each sensor becomes informed in any given slice. The upper bound given in Eq.~\eqref{41} is the worst case for a slice, as this case is likely to happen given the possibility of any arbitrary network configuration. Drawing an analogy with the LTI scenario, the slices have to be completed infinitely often and thus, Theorem~\ref{thm0} considers all of the slices and provides different ways to guarantee an infinite number of slices. We emphasize that information diffusion in the network can actually deteriorate (not necessarily in an order) and Theorem~\ref{thm0} further provides the ``rate'' at which well-behaved network configurations must occur. Since the discussion so far mostly caters to ``forgetting the sensor initial conditions'', i.e. the asymptotic stability, we now show the steady-state of the DDF for a particular application of interest. 

\subsection{Dynamic Leader-Follower}
In this setup, the goal for the entire sensor network is to converge to the state of one anchor (multiple anchors and converging to their linear-convex combination may also be considered, see e.g.~\cite{khan2009distributed,khan2010diland}). Let~${\bf{1}}_{n}$ be the~$n \times 1$ column vector of~$n~1$’s, and~$u$ be the scalar state of the (single) anchor, which is known and does not change over time. The leader-follower algorithm requires~$\lim_{k \rightarrow \infty}{\bf{x}}(k) \nonumber = {\bf{1}}_{n}{{u}}$, where $\mb{x}_k\in\mbb{R}^n$ collects the states of all of the sensors. Since this is a dynamic algorithm with mobile sensors, a sensor may not find any neighbor at many time instants. When a sensor does find neighbors, an anchor may not be one of them. Furthermore, if a sensor has the anchor as a neighbor at some time, this anchor may not be a neighbor going forward because the nodes are mobile. We now use the results from Section~\ref{stability}, to provide the asymptotic stability analysis of the dynamic leader-follower algorithm. 
\begin{thm}
Consider a network of~$n$ sensors and~$s=1$ anchor with the following update:
\begin{equation}\label{71}
{\bf{x}}(k+1)={{P}}_k{\bf{x}}(k)+{{B}}_k{{u}},\qquad k \geq 0,
\end{equation}
in which~$u$ is the state of the anchor. With assumption {\bf{A0-A2}}, in addition to the following
\begin{equation}\label{lf}
\sum_j ({P}_{k})_{i,j}+({B}_{k})_{i,j}=1, \qquad \forall k,
\end{equation}
all sensors (asymptotically) converge to the anchor state.
\end{thm}

\begin{proof}
It can be verified that Eq.~\eqref{71} results into
\begin{eqnarray*}
{\bf{x}}(k+1)= (P_k \ldots P_0){\bf{x}}(0) + \sum\limits_{m=0}^{k}\left(\prod_{j=1}^m P_{k-j+1}\right)B_{k-m}{{u}}.\\
\end{eqnarray*}
With the slice notation, we have
\begin{equation}\label{72}
{\bf{y}}(t+1)={{M}}_t{\bf{y}}(t)+{{N}}_t{{u}},\qquad k \geq 0,
\end{equation}
where~${\bf{y}}(0)={\bf{x}}(0),M_0 = P_{\vert M_0 \vert -1} \ldots P_0,$ and
\begin{eqnarray}
M_t &=& P_{\left(\sum\limits_{i=0}^{t}{\vert M_i \vert}\right) - 1} \ldots P_{\left(\sum\limits_{i=0}^{t-1}{\vert M_i \vert}\right)},~ t >0,\label{73}\\
N_t&=&\sum\limits_{m=0}^{{\vert M_t \vert - 1}}\left(\prod_{j=1}^m P_{{\vert M_t \vert}-j}\right)B_{{\vert M_t \vert - 1}-m}\label{74}.
\end{eqnarray}
In addition, we have
\begin{eqnarray*}
{\bf{y}}(t+1)=(M_t \ldots M_0){\bf{y}}(0) + \sum\limits_{m=0}^{t}\left(\prod_{j=1}^m M_{t-j+1}\right)N_{t-m}{{u}}.
\end{eqnarray*}
Since~$u$ is a constant, and
\begin{eqnarray}\label{rho}
\rho(M_t)\leq {\Vert M_t \Vert}_{\infty} < 1, \qquad \forall t,
\end{eqnarray}
as~$t \rightarrow \infty$ in Eq.~\eqref{72},~${\bf{y}}_{t+1}$ converges to a limit,~${\bf{y}}^{*}$. This limit is further unique because it is a fixed point of a linear iteration with bounded matrices,~\cite{tsit_book}. Therefore,
\begin{equation*}
\lim\limits_{t \rightarrow \infty} {\bf{y}}_{t+1}={\bf{y}}^{*},
\end{equation*}
and we have
\begin{eqnarray}
{\bf{y}}^{*} = M_t {\bf{y}}^{*} + N_t {{u}} \rightarrow (I_n-M_t){\bf{y}}^{*} = N_t {{u}},
\end{eqnarray}
where~${\bf{y}}^{*}={\bf{x}}^{*}$ is the limiting states of the sensors. Thus,
\begin{eqnarray}
{\bf{x}}^{*} = {(I_n-M_t)}^{-1} N_t {{u}},
\end{eqnarray}
for which we used the fact that~${(I-M_t)}$ is invertible due to~Eq.~\eqref{rho}. In order to show that the limiting states of the sensors are indeed the anchor state, we require
\begin{eqnarray}\label{80}
{(I_n-M_t)}^{-1} N_t = \mb{1}_n~~\Rightarrow~~M_t{\bf{1}}_n+N_t = {\bf{1}}_n.
\end{eqnarray}
Note that~$N_t$ is a column vector since there is only one anchor. Before we proceed, for the sake of simplicity let us represent any arbitrary $t$th slice as:
\begin{eqnarray*}
M_t &\triangleq& P_{T} P_{{T}-1} \ldots P_{0}, \qquad \vert M_t \vert =T+1.
\end{eqnarray*}
By substituting~$M_t$ and~$N_t$ from Eqs.~\eqref{73} and~\eqref{74} in Eq.~\eqref{80}, we need to show that
\begin{eqnarray}
( P_{T} \ldots P_{{0}}){\bf{1}}_n+
\sum\limits_{m=0}^{T}\left(\prod_{j=1}^m P_{T+1-j}\right)B_{T-m} ={\bf{1}}_n.
\end{eqnarray}
By expanding the left hand side of the above, we have
\begin{eqnarray}
(P_{T} P_{{T}-1}\ldots P_{{0}}){\bf{1}}_n&+&(P_{T} P_{{T}-1}\ldots P_{{1}})B_{0}\nonumber\\
&+&(P_{T} P_{{T}-1}\ldots P_{{2}})B_{1}\nonumber\\
&\vdots&\nonumber\\
&+&(P_{T} P_{{T}-1})B_{T-2}\nonumber\\
&+&(P_{T})B_{T-1}\nonumber\\
&+&(B_{T}).\label{82}
\end{eqnarray}
The first line of the above expression can be simplified as
\begin{eqnarray}\label{7-2}
(P_{T} P_{{T}-1}\ldots P_{{1}})( P_0{{\bf{1}}_n}+B_{0}),
\end{eqnarray}
in which~$B_{0} \neq 0$ is a~$n \times 1$ vector corresponding to the first sub-stochastic update at the beginning of the slice,~$M_t$. Also,~$B_{0}$ has only one non-zero, say~$\alpha_i$, at the~$i$th position if sensor~$i$ updates with the anchor at the beginning of the slice,~$M_t$. From Eq.~\eqref{lf}, it can be verified that
\begin{eqnarray}
P_{0}{\bf{1}}_n+B_{0}={\bf{1}}_n.
\end{eqnarray}
Therefore, Eq.~\eqref{82} reduces to
\begin{eqnarray}\label{85}
(P_{T} P_{{T}-1}\ldots P_{{1}}){\bf{1}}_n+(P_{T}\ldots P_{{2}})B_{{1}}+ \ldots+B_{T}.
\end{eqnarray}
After the first (sub-stochastic) update, each~$B_{j}, (1\leq j\leq T)$, has exactly one non-zero in case of sub-stochastic updates and all zeros otherwise. The procedure continues in a similar way for any sub-stochastic update, i.e. update with the anchor. Let us consider now the alternate case where the update is stochastic, i.e. without the anchor and with some neighboring sensors. Suppose~$B_{c}$ is the next sub-stochastic update, and we have~$B_{j}=\mb{0}_{n-1}, (1 \leq j < c)$. Eq.~\eqref{85} then reduces to
\begin{eqnarray}
(P_{T}\ldots P_{{c}+1})(P_{c} P_{{c}-1}\ldots P_{{1}}
{\bf{1}}_n+B_{{c}})+\ldots+(B_{T}).\label{86}
\end{eqnarray}
Since between~$P_1$ and~$P_c$ there is no sub-stochastic update,~$P_{{c}-1}\ldots P_{{1}}
{\bf{1}}_n={\bf{1}}_n$, and we can rewrite Eq.~\eqref{86} as
\begin{eqnarray}
(P_{T}\ldots P_{{c}+1})(P_{c}
{\bf{1}}_n+B_{{c}})+\ldots+(B_{T}),\label{87}
\end{eqnarray}
and the procedure continues as before (note the similarity between Eq.~\eqref{7-2}, and the first term on the left hand side of Eq.~\eqref{87}). Finally, 
\begin{eqnarray}
( P_{T} \ldots P_{{0}}){\bf{1}}_n&+&
\sum\limits_{m=0}^{M}(\prod_{j=1}^m P_{T+1-j})B_{T-m}\nonumber\\
&=&P_{{T}}{\bf{1}}_n+B_{T}\nonumber\\
&=&{\bf{1}}_n,
\end{eqnarray}
which leads to $\lim\limits_{k \rightarrow \infty}{\bf{x}}(k)= {\bf{x}}^{*}={{u}}$.
\end{proof}

\section{Illustrative Example}\label{example}
In this section, we provide a few numerical examples to illustrate the concepts described in this paper. We show the product of~$4 \times 4$ (sub-) stochastic matrices. Assumptions {\bf{A0-A2}} are satisfied with~$\beta_1=0.05$,~$\beta_2=0.7$. At each iteration, the update matrix, which is left multiplied to the product of past matrices, randomly takes one of the following forms: 
(i) identity matrix except for the~$i$th~($1\leq i \leq4$) row, which is replaced by a stochastic row vector; or,
(ii) identity matrix except for the~$i$th~($1\leq i \leq4$) row, which is replaced by a sub-stochastic row vector; or, 
(iii) a~$4 \times 4$ identity matrix,~$I_4$. 
Fig.~\ref{f2} (Left) shows the infinity norm and the spectral radius of the product of system matrices. In addition, the infinity norm of the product of slices are illustrated for comparison. Slice lengths, at the termination of each slice and over the slice index,~$t$, are shown in Fig.~\ref{f2} (Right). The minimum slice length is $5$, and $15$ slices are completed within $200$ iterations of this simulation. Note that the infinity norm of the slices is the only (strictly) monotonically decreasing curve. 
\begin{figure}[!h]
\centering
\subfigure{\includegraphics[width=1.72in,height=1.70in]{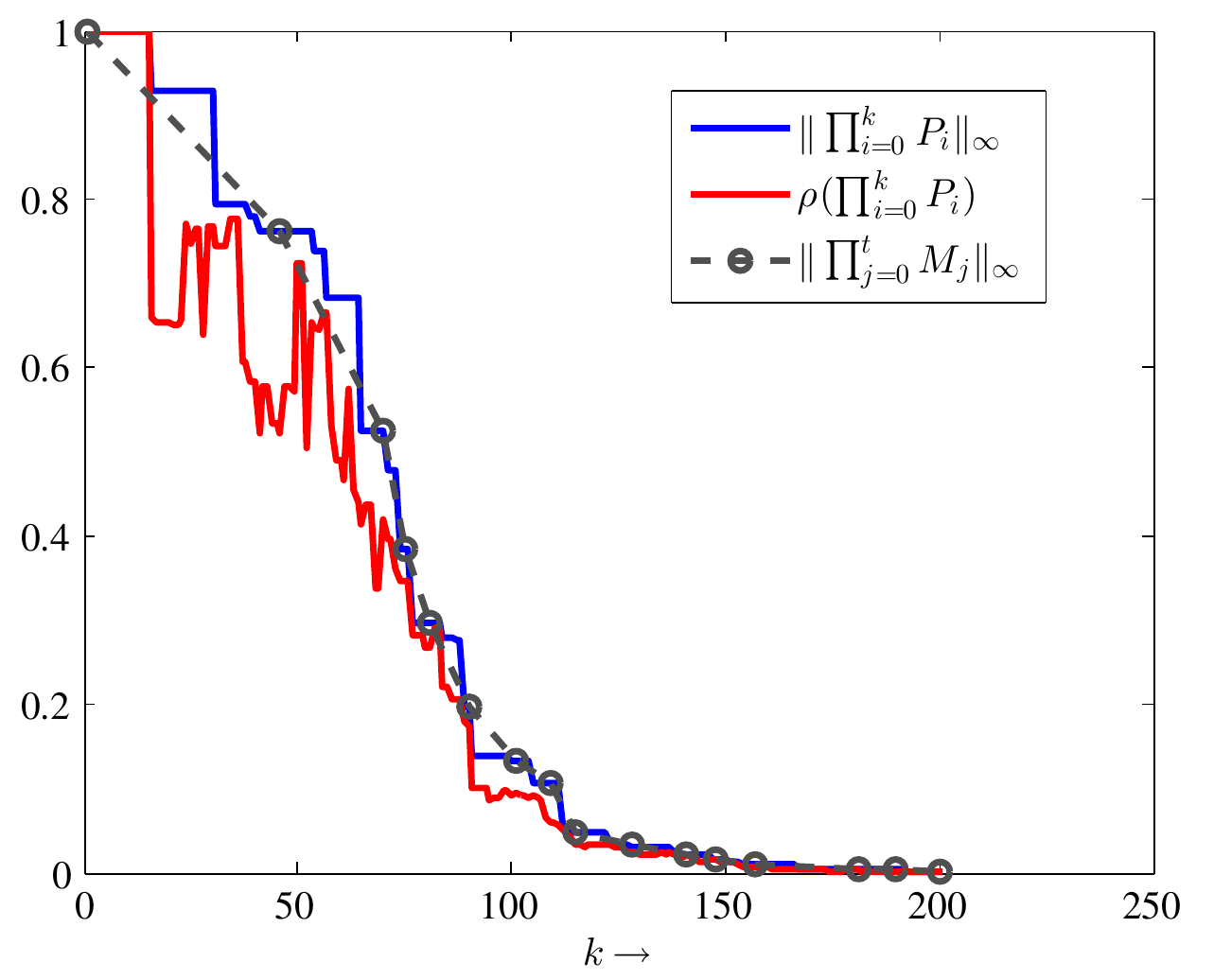}}
\subfigure{\includegraphics[width=1.72in,height=1.70in]{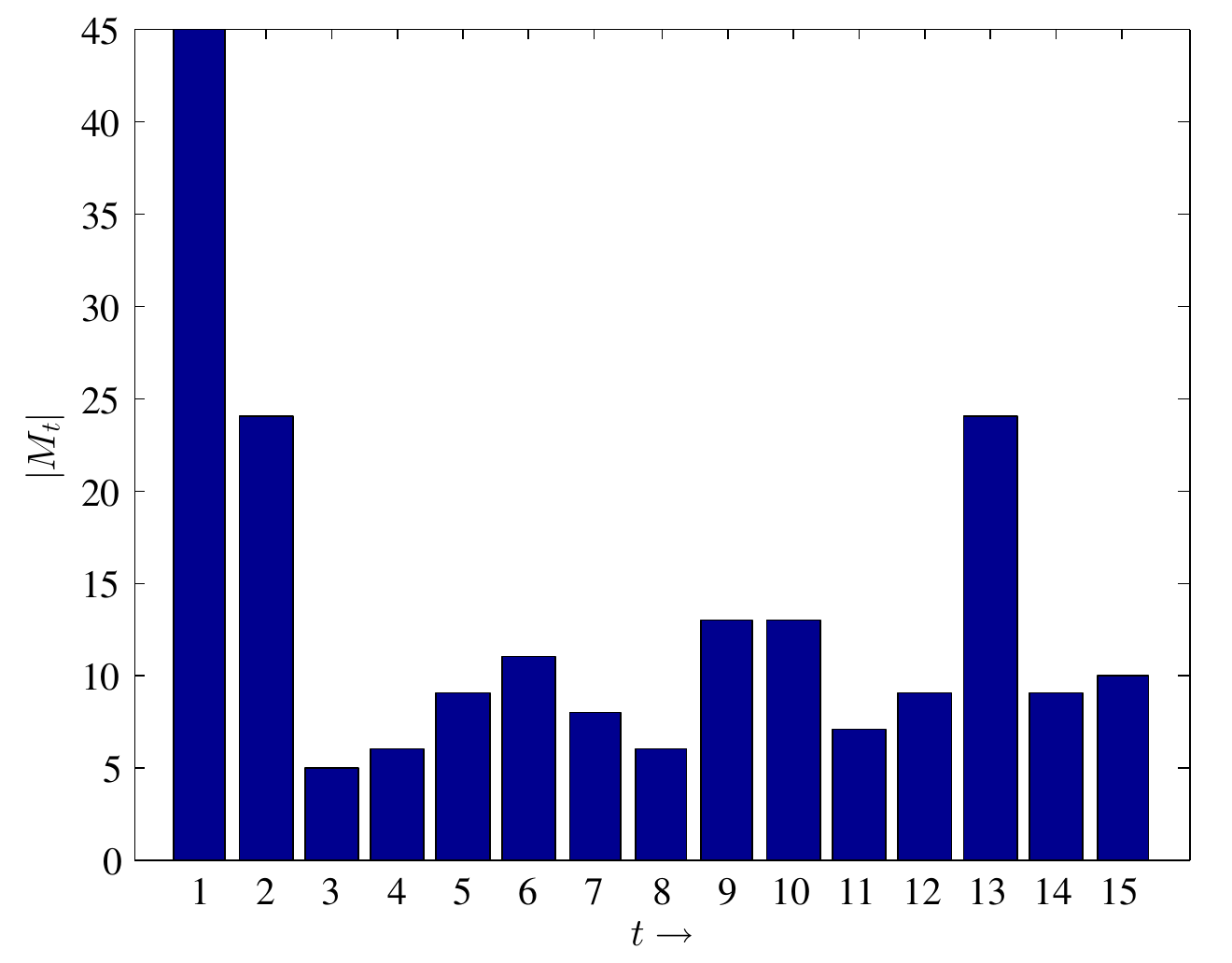}}
\caption{(Left) Spectral Radius vs. Infinity Norm. (Right) Slice lengths.}
\label{f2}
\end{figure}

In Figs.~\ref{f3} and~\ref{f4}, we illustrate the dynamic leader-follower algorithm. Fig.~\ref{f3} shows the network configuration with~$n=4$ mobile sensors, where sensor~$i$ is restricted to move in the region,~$R_i$, marked as the corresponding disk. The anchor only moves in the region,~$R_0$, and the random trajectories taken by each node are marked; shown only over the first~$40$ iterations to maintain visual clarity as random trajectories clutter in a short time. We choose~$1.5$ times the radius of the innermost circle as the communication radius; note that only sensors,~$1,2$ in regions~$R_1,R_2$, may be able to talk to the anchor given this communication radius and depending on the corresponding node locations within the respective regions,~$R_0,R_1,R_2$, see the top-left figure. In the top-right figure, we show a time instant when no sensor is able to communicate with any other node; thus resulting in an identity system matrix. The bottom-left figure shows the case when only one sensor,~$1$ in region~$R_1$, communicates with the anchor; thus resulting in a sub-stochastic system matrix. Finally, the bottom-right figure shows the stochastic update when sensor~$3$, in region~$R_3$, is able to communicate with sensor~$4$, in region~$R_4$. Clearly, we have chosen this network configuration, (random) motion model, and communication radius for visual convenience; the setup is applicable to any scenario where the communication radius and random motion models ensure that the information (possibly over a longer time window) travels from the anchor to each mobile sensor. 
\begin{figure}[!h]
\centering
\subfigure{\includegraphics[width=1.72in]{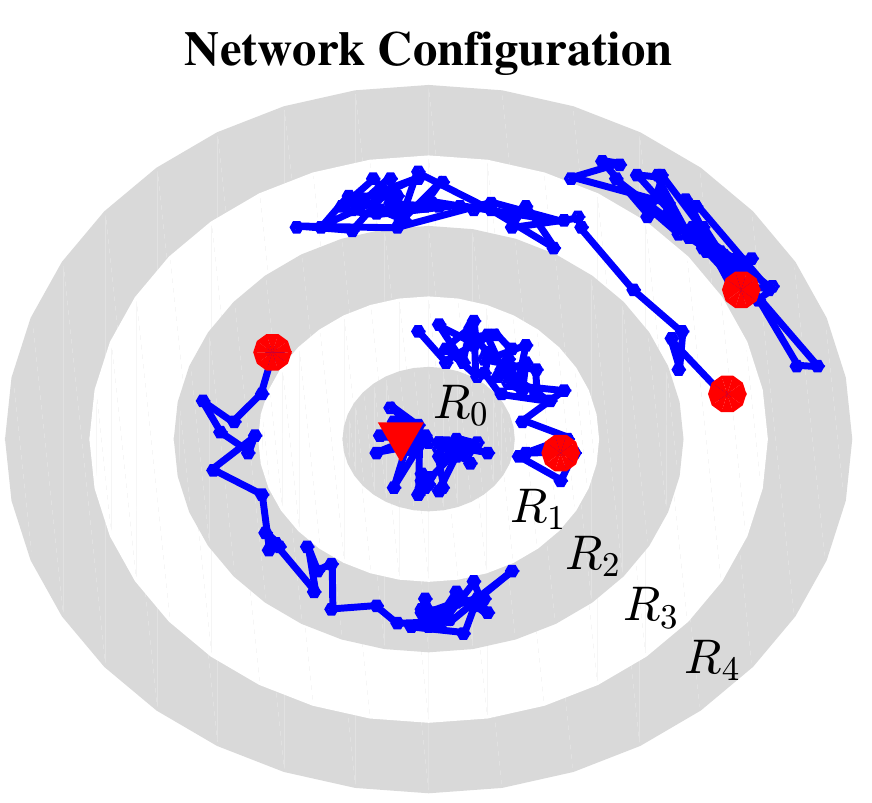}}
\subfigure{\includegraphics[width=1.72in]{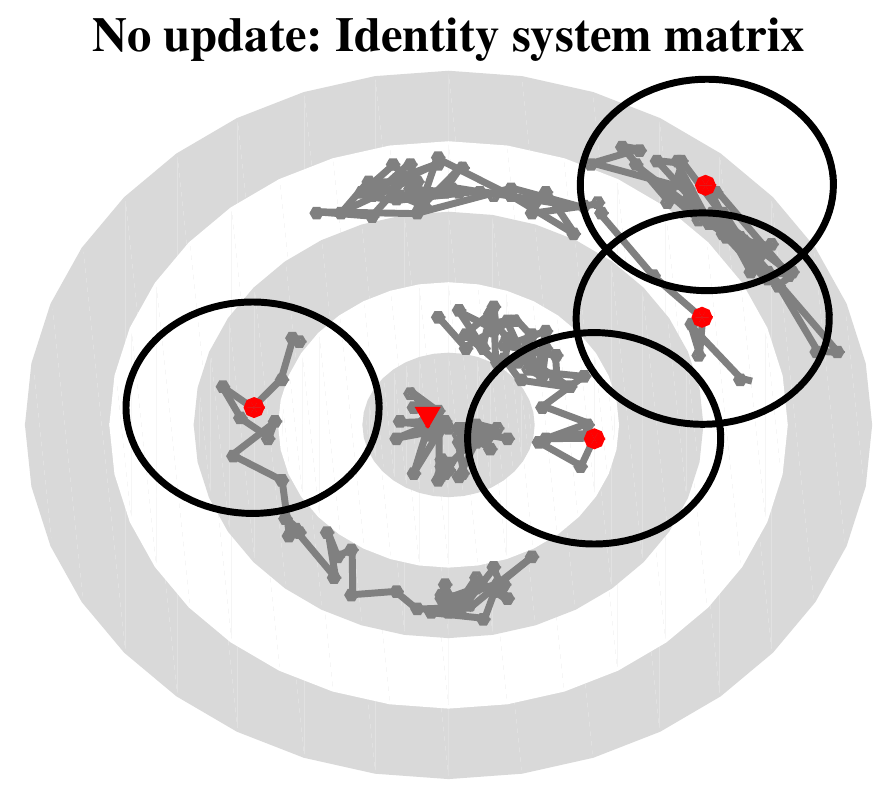}}
\subfigure{\includegraphics[width=1.72in]{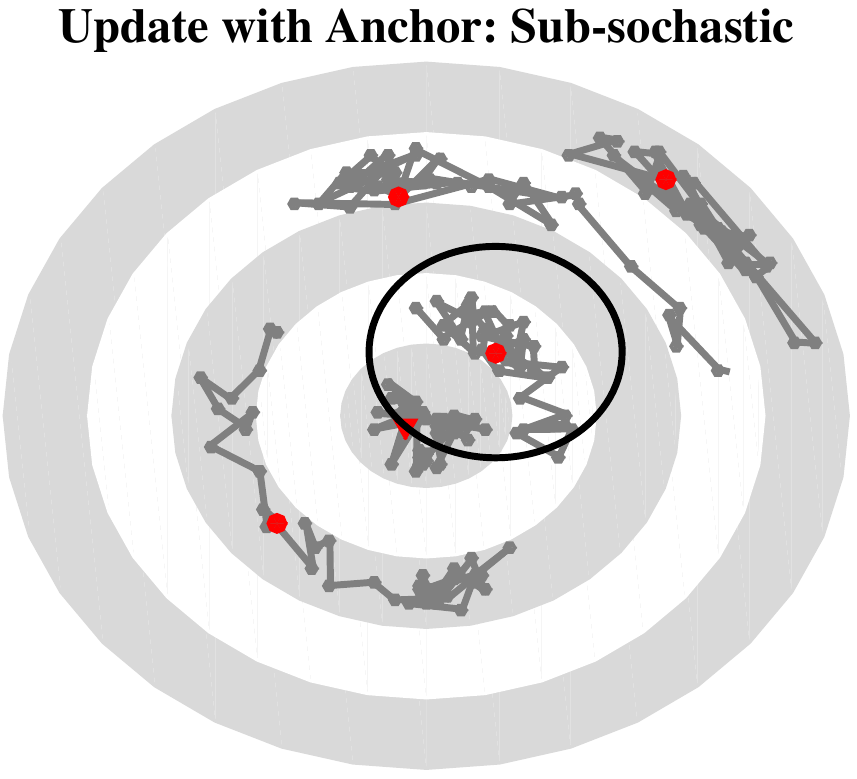}}
\subfigure{\includegraphics[width=1.72in]{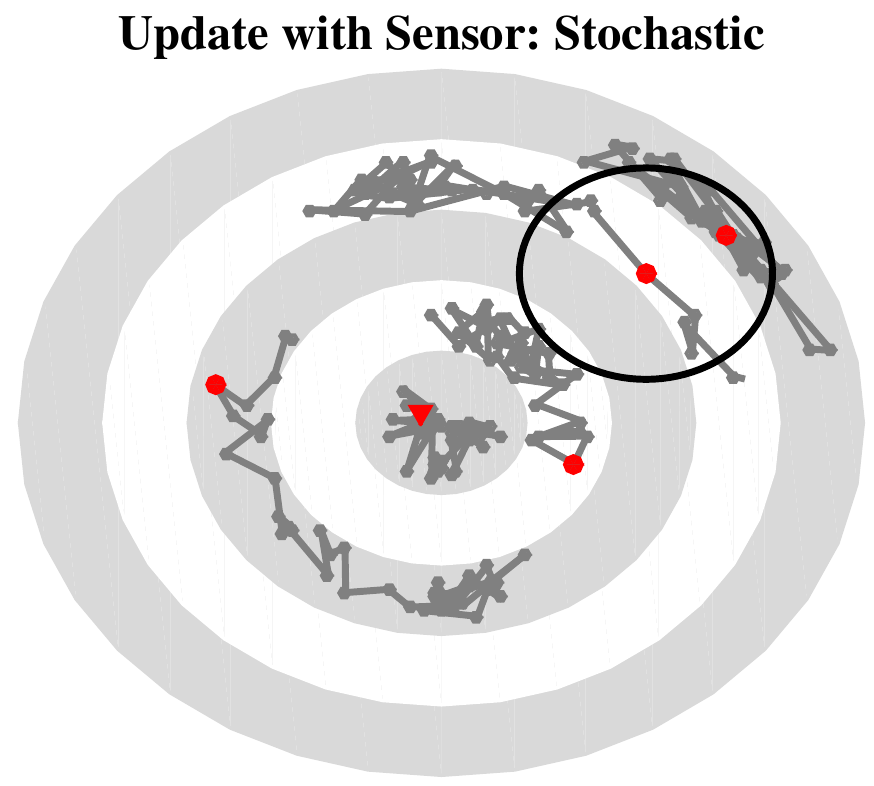}}
\caption{Dynamic leader-follower: Mobile sensors, red circles, and the anchor, red triangle, follow a restricted motion in their corresponding disks. The blue (and gray) lines show the nodal trajectories whereas the circles around the sensors show their communication radii.}
\label{f3}
\end{figure}

Finally, Fig.~\ref{f4} shows the sensor states with the anchor state chosen at~$u=3$. We note that the sensors closer to the anchor converge faster to the anchor state as compared to the farther sensors. This is because of the information flow in this particular scenario. That a sensor, whose state is closer to the anchor state, does not lose this information is ensured by the conditions established on the sensor weights. In particular, an informed sensor does not lose its (partial) knowledge when updating only with neighboring sensors because: (i) each sensor assigns some weight to its past information; and (ii) no sensor is allowed to assign an arbitrarily large weight on any neighboring sensor. We emphasize that this simple illustration is significantly insightful and demonstrates the key concepts of the theoretical results described in this paper. Clearly, the setup can be extended to arbitrary motion models, network configurations, and large networks. 
\begin{figure}[!h]
\centering
\includegraphics[width=2.65in]{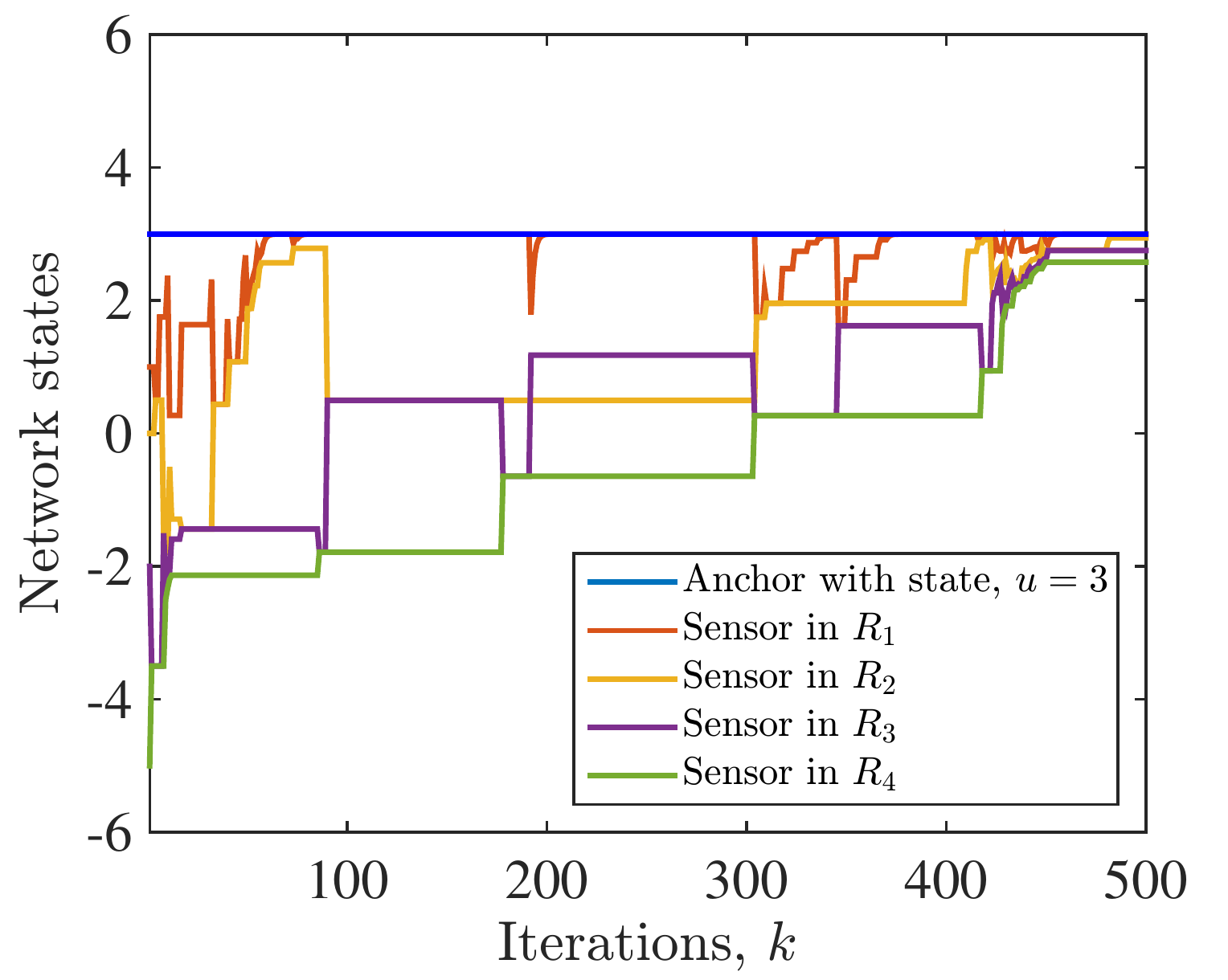}
\caption{Dynamic leader-follower: Sensor and anchor states.}
\label{f4}
\end{figure}

\section{Conclusion}\label{conc}
In this paper, we study asymptotic stability of Linear Time-Varying (LTV) systems with (sub-) stochastic system matrices. Motivated by applications in distributed dynamic fusion (DDF), we design the conditions on the system matrices that lead to asymptotic stability of such dynamics. Rather than exploring the joint spectral radius of the (infinite) set of system matrices, we partition them into non-overlapping slices, such that each slice has a subunit infinity norm, and slices cover the entire sequence of the system matrices. We use infinity norm to characterize the asymptotic stability and provide upper bounds on the infinity norm of each slice as a function of the slice length and some additional system parameters. We show that asymptotic stability is guaranteed not only in the trivial case where all (or an infinite subset) of slices have a bounded length, but also if there exist an infinite subset of slices whose (unbounded) lengths do not grow faster than a particular exponential growth. We apply these theoretical findings to the dynamic leader-follower algorithm and establish the conditions under which each sensor converges to the state of the anchor. These concepts are further illustrated with insightful examples.

\bibliographystyle{IEEEtran}
\bibliography{bibliography}
\end{document}